\documentclass{article}

\usepackage{microtype}
\usepackage{graphicx}
\usepackage{subcaption}
\usepackage{booktabs} 

\usepackage{hyperref}

\usepackage[preprint]{icml2026}

\usepackage{amsmath}
\usepackage{amssymb}
\usepackage{mathtools}
\usepackage{amsthm}
\usepackage{comment}
\newcommand{\G}{\mathcal{G}}
\usepackage{dsfont} 

\newcommand{\E}{\mathbb{E}}
\newcommand{\R}{\mathbb{R}}
\DeclareMathOperator{\Var}{Var}

\usepackage[capitalize,noabbrev]{cleveref}

\theoremstyle{plain}
\newtheorem{theorem}{Theorem}[section]

\newtheorem{lemma}[theorem]{Lemma}
\newtheorem{corollary}[theorem]{Corollary}
\theoremstyle{definition}

\newtheorem{assumption}[theorem]{Assumption}
\theoremstyle{remark}
\newtheorem{remark}[theorem]{Remark}
\allowdisplaybreaks
\usepackage[textsize=tiny]{todonotes}

\icmltitlerunning{Individualized Causal Effects under Network Interference with Combinatorial Treatments}

\begin{document}

\twocolumn[
  \icmltitle{Individualized Causal Effects under Network Interference with Combinatorial Treatments}



  \icmlsetsymbol{equal}{*}

  \begin{icmlauthorlist}
    \icmlauthor{Yunping Lu}{yyy}
    \icmlauthor{Haoang Chi}{comp}
    \icmlauthor{Qirui Hu}{sch}
    \icmlauthor{Zhiheng Zhang}{sch}
   
  \end{icmlauthorlist}

  \icmlaffiliation{yyy}{University of Leeds}
  \icmlaffiliation{comp}{National University of Defense Technology}
  \icmlaffiliation{sch}{School of Statistics and Data Science, SUFE}

  \icmlcorrespondingauthor{Zhiheng Zhang}{zhangzhiheng@mail.shufe.edu.cn}

  \icmlkeywords{Machine Learning, ICML}

  \vskip 0.3in
]



\printAffiliationsAndNotice{}  

\begin{abstract}
Modern causal decision-making increasingly demands individualized treatment-effect estimation in networks where interventions are high-dimensional, combinatorial vectors. While network interference, effect heterogeneity, and multi-dimensional treatments have been studied separately, their intersection yields an exponentially large intervention space that makes standard identification tools and low-dimensional exposure mappings untenable. We bridge this gap with a unified framework that constructs a \emph{global potential-outcome emulator} for unit-level inference. Our method combines (1) rooted network configurations to leverage local smoothness, (2) doubly robust orthogonalization to mitigate confounding from network position and covariates, and (3) sparse spectral learning to efficiently estimate response surfaces over the $2^p$-dimensional treatment space. We also decompose networked effects into own-treatment, structural, and interaction components, and provide finite-sample error bounds and asymptotic consistency guarantees. Overall, we show that individualized causal inference remains feasible in high-dimensional networked settings without collapsing the intervention space.
\end{abstract}

\section{Introduction}
Causal inference in networked systems has become a central problem across social sciences, epidemiology, and online experimentation. In many applications, units are embedded in a network, and interventions are high-dimensional \emph{treatment slates}—vectors of features that can be simultaneously toggled. Examples include vaccination campaigns, peer effects in education, and product changes on digital platforms. In these settings, the classical no-interference assumption (SUTVA) is often violated: one unit's outcome may depend on others' treatment assignments, making the causal estimand inherently \emph{global} in the assignment vector \citep{hudgens2008toward, aronow2017estimating, savje2021average, leung2022approximate}.

At the same time, policy decisions are increasingly individualized: practitioners seek \emph{unit-level} or \emph{conditional} effects to guide personalized targeting and adaptive experimentation \citep{chernozhukov2018double, kunzel2019metalearners, wager2019grf}. Reconciling interference, heterogeneity, and multi-dimensional interventions is therefore crucial for credible causal decision-making in networked environments.

This paper integrates three literatures that have advanced in parallel. First, the interference literature has developed estimands and randomization-based methods for average direct and spillover effects \citep{hudgens2008toward, aronow2017estimating, savje2021average}. Second, the heterogeneous treatment effect (HTE) literature provides tools for learning individualized effects under no interference, ranging from machine learning to meta-learners \citep{chernozhukov2018double, kunzel2019metalearners, wager2019grf}, extended to interference settings \citep{ma2022learning, agarwal2022network}. Third, factorial and multi-treatment causal inference has clarified potential-outcome formulations for multi-factor interventions and randomization-based inference \citep{dasgupta2015factorial, lopez2017multiple, zhao2021factorial, agarwal2023synthetic}.

These three strands differ in their \emph{primary bottleneck}: interference requires modeling dependence between units’ assignments, HTE demands controlling confounding while preserving heterogeneity, and multi-dimensional treatments face another conceptual exponentially large action space. Real-world deployments often encounter all three challenges simultaneously. Crucially, this is not a simple additive problem in which interference, heterogeneity, and multi-dimensional treatments can be handled separately. Their interaction creates endogenous dependence structures and configuration-specific overlap conditions, under which standard identification and estimation arguments from each individual literature no longer apply. Consider a large online platform running a networked experiment. Each user $i\in[N]$ is assigned a $p$-dimensional binary feature vector (a `slate'') $T_i\in\{-1,1\}^p$, and the platform wishes to answer individualized questions of the form: \emph{`For this user, what is the outcome change if we toggle a particular subset of features, under the user's current local network environment?''} In principle, even without interference, this requires learning a response surface over $2^p$ treatment combinations; with interference, the relevant counterfactual depends on neighbors' assignments and the local network configuration, producing a space of counterfactuals that is exponentially large in both $p$ and network size. Existing approaches fail in a common, structural way.

This paper introduces a framework that transforms this seemingly intractable problem into one amenable to high-dimensional inference. Our approach combines: (i) graph localization to represent interference environments with rooted network configurations, exploiting smoothness in configuration space; (ii) doubly-robust (orthogonal) residualization to mitigate confounding from both network position and covariates; and (iii) sparse spectral learning for high-dimensional combinatorial treatments, enabling stable estimation and extrapolation. These components together yield a global potential-outcome emulator that produces individualized causal contrasts for arbitrary assignments.

A key advantage of the framework is its principled decomposition of individualized causal effects into: (a) own-treatment effects, (b) structural effects from changes in the network configuration, and (c) joint effects that combine both. This decomposition clarifies which components of heterogeneity are identifiable from network data and aligns with practical interventions such as feature toggles and network redesigns. We show that, under local regularity and sparse spectral structure, individualized causal comparisons remain feasible even when the counterfactual space grows exponentially with both the treatment dimension and network size.

The main contributions of this paper are as follows:
\begin{itemize}\setlength\itemsep{2pt}
    \item We formulate a potential-outcome framework for \emph{individualized} causal effects under network interference with \emph{combinatorial} treatments, with a clear decomposition into own-treatment, structural, and joint effects.
    \item We develop an estimator that integrates graph-configuration localization, doubly robust orthogonalization, and sparse spectral learning to recover unit-level response functions and construct a global potential-outcome emulator.
    \item We establish both finite-sample error bounds and asymptotic guarantees under localized dependence and overlap, and characterize robustness and partial identification when some treatment directions are not locally identifiable.
\end{itemize}

\section{Problem Formalization}
\label{sec:model}

Let $\mathbf{t}=(t_1,\ldots,t_N)$ denote a \emph{global} assignment, where each unit $i$ receives a $p$-dimensional binary treatment slate $t_i\in\{-1,1\}^p$.
For each $\mathbf{t}$, let $Y_i(\mathbf{t})$ denote the potential outcome of unit $i$, and denote $Y_i=Y_i(\mathbf{T})$ as the observed outcome under actual assignment $\mathbf{T}$.
We observe a single assignment--outcome snapshot on a networked population,
\begin{align}
    \bigl\{(Y_i,\,T_i,\,X_i,\,\mathcal{N}_i): i\in[N]\bigr\},
\end{align}
where $X_i$ are pre-treatment covariates and $\mathcal{N}_i$ encodes network neighborhood information.

Our goal is to construct an estimator $\widehat{\mathbf{Y}}(\mathbf{t})=(\widehat Y_1(\mathbf{t}),\ldots,\widehat Y_N(\mathbf{t}))$ that approximates the vector of potential outcomes under arbitrary global assignments.
This problem is well defined but intrinsically challenging: the assignment space grows exponentially in $p$ and $N$, and interference induces dependence across units through the network.

Beyond global predictions, we focus on \emph{individualized} causal inference.
For each unit, we target three families of causal contrasts: \emph{own-treatment} effects that vary the unit’s own slate holding the local interference environment fixed, \emph{structural} effects that vary the local network configuration holding the slate fixed, and \emph{joint} effects that vary both.
For user-specified contrasts, we provide point estimates with uncertainty quantification, while retaining the ability to evaluate $\widehat{\mathbf{Y}}(\mathbf{t})$ for arbitrary $\mathbf{t}$.

Following \citet{auerbach2021local}, we summarize the interference environment of unit $i$ by a rooted \emph{network configuration} $G_i(\mathbf{t})$, constructed from an ego-centered neighborhood of $i$ (e.g., within a fixed graph radius) together with vertex marks that include treatment slates (and, if desired, additional marks such as discretized covariates).
We write $G_i := G_i(\mathbf{T})$ for the realized random configuration and $g_i$ for its realization.
Let $\mathcal{G}$ denote the space of rooted, marked configurations (up to root-preserving isomorphism).

We equip $\mathcal{G}$ with a truncated rooted-graph distance (as in the local approach of \citet{auerbach2021local}).
Let $B_r(g)$ denote the radius-$r$ rooted ball around the root of $g$ (induced subgraph with marks carried along).
When comparing $g$ and $g'$, we first check whether $B_r(g)$ and $B_r(g')$ are root-isomorphic \emph{as unmarked graphs}; among root-preserving isomorphisms, we then measure mark mismatches.
Concretely, let $\tau_g(v)\in\{-1,1\}^p$ denote the treatment-slate mark at vertex $v$ in configuration $g$, and define $\Delta_r(g,g') := \displaystyle \min_{\phi:\,B_r(g)\simeq B_r(g')}
{\sum_{v\in V(B_r(g))}
\mathbb{I}\bigl\{\tau_g(v)\neq \tau_{g'}(\phi(v))\bigr\}}/|V(B_r(g))|$ if root-isomorphic (otherwise equals to $1$)
We then truncate at a small radius $R\in \mathbb{Z}^{+}$ and set
\begin{equation}\label{eq:d_R}
d_R(g,g') := \sum_{r=0}^{R} 2^{-(r+1)}\,\Delta_r(g,g').
\end{equation}
The metric $d_R$ induces a notion of ``local similarity'' between interference environments and enables nonparametric localization over $\mathcal{G}$.

\paragraph{Illustrative example.}
To illustrate the rooted-graph distance and the role of the matching $\phi$, consider a social network where each unit receives a binary treatment slate and outcomes may depend on neighbors’ assignments. Let $R=1$, thus each configuration is the root and its immediate neighbors. Suppose the root has exactly two neighbors, and consider two configurations $g$ and $g'$ whose unmarked ego networks are identical: in both, the root is connected to two neighbors. Hence $B_1(g)$ and $B_1(g')$ are root-isomorphic as unmarked graphs. A root-preserving isomorphism $\phi$ is then a bijection between the neighbor sets of $g$ and $g'$ that fixes the root. When multiple such bijections exist, $\Delta_1(g,g')$ uses the one that minimizes treatment-slate mismatches. In $g$, the two neighbors receive slates $(+1,-1)$ and $(-1,+1)$; in $g'$, they receive $(+1,+1)$ and $(-1,+1)$. Under the optimal matching, one neighbor is paired with an identical slate, while the other is paired with a different slate, so there is exactly one mismatch and $\Delta_1(g,g')=1/2$. If instead $g'$ had a different local structure—e.g., a different number of neighbors—then $B_1(g)$ and $B_1(g')$ would not be root-isomorphic, no root-preserving matching would exist, and by definition $\Delta_1(g,g')=1$. The overall distance $d_R(g,g')$ then aggregates these discrepancies across radii with geometrically decaying weights, placing more emphasis on mismatches closer to the root, consistent with the idea that nearby interference environments matter most for the root’s outcome.

Also, to represent response surfaces over $\{-1,1\}^p$, we use Walsh characters.
For each subset $S\subseteq[p]$, define $Z_S(t)=\prod_{\ell\in S} t_\ell$ and let $Z(t)$ collect $\{Z_S(t):S\subseteq[p]\}$.
This provides a convenient orthonormal dictionary on the hypercube and supports sparse/near-sparse modeling of high-order interactions.

Given a target unit $i$, we localize around its realized configuration $g_i$ by assigning weights to sample units $j$:
\begin{equation}\label{eq:weights}
w^{(i)}_j
:= \frac{K_G\!\bigl(d_R(g_j,g_i)/b_G\bigr)}
{\sum_{k=1}^N K_G\!\bigl(d_R(g_k,g_i)/b_G\bigr)},
\qquad \sum_{j=1}^N w^{(i)}_j = 1,
\end{equation}
where $K_G$ is a compactly supported kernel (e.g., indicator or Epanechnikov) and $b_G>0$ is a bandwidth (or, equivalently, one may use a $k$NN radius).
We quantify the \emph{effective} local sample size via the Kish measure
$
n_{\mathrm{eff}}^{(i)} := {1}/{\sum_{j=1}^N \bigl(w^{(i)}_j\bigr)^2},
$
which controls the variance of localized averages under homoskedastic noise and will serve as the fundamental local sample-size parameter in our analysis.


The following assumptions make the target problem well-defined and render it statistically tractable; each is standard in nearby literatures and will be revisited in the identification and theory sections.

\begin{assumption}[Configuration sufficiency / local interference]\label{ass:local}
There exists a (possibly radius-truncated) configuration mapping $G_i(\mathbf{t})$ such that, for all $i$ and all global assignments $\mathbf{t}$, the potential outcome $Y_i(\mathbf{t})$ depends on $\mathbf{t}$ only through the unit's own slate $t_i$ and its rooted configuration $G_i(\mathbf{t})$ (and covariates $X_i$), and thus $Y_i(\mathbf{t}) \;=\; Y_i\!\bigl(t_i,\,G_i(\mathbf{t})\bigr)$.
\end{assumption}
Assumption~\ref{ass:local} reduces the global interference problem to a localized object in $\mathcal{G}$, enabling borrowing of information across units with similar environments \citep{auerbach2021local}. Sensitivity to the truncation radius $R$ (and to alternative configuration constructions) can be assessed empirically by checking stability of fitted effects as $R$ varies.
\emph{Relaxations.} If interference decays with graph distance, increasing $R$ yields a controlled bias--variance trade-off; our methods naturally accommodate such sensitivity analyses.

\begin{assumption}[Local regularity in configuration space]\label{ass:smooth}
For each fixed $(t,x)$, the function
$g \mapsto \mathbb{E}[Y(t,G)\mid G=g, X=x]$
is locally Lipschitz with respect to $d_R$.
Specifically, $\forall g_0\in\mathcal{G}$ there exists $L(g_0,x,t)>0$ s.t.
$\big|
\mathbb{E}[Y(t,G)\mid G=g, X=x]
-
\mathbb{E}[Y(t,G)\mid G=g_0, X=x]
\big|
\;\le\;
L(g_0,x,t)\, d_R(g,g_0)
$
for all $g$ in a neighborhood of $g_0$.
\end{assumption}

Assumption~\ref{ass:smooth} justifies kernel or $k$NN localization and ensures that localization bias vanishes as neighborhoods shrink. One may diagnose violations by examining residual stability as a function of $d_R(g_j,g_i)$ and by cross-validating localization bandwidths. Piecewise regularity can be handled by using covers/adaptive bandwidths; when smoothness fails in specific regions, our framework permits reporting localized uncertainty inflation or reverting to coarser, partially identified summaries.

\begin{assumption}[Local overlap / design richness]\label{ass:overlap}
Fix a target configuration--covariate pair $(g_0,x_0)$.
There exists a neighborhood $\mathcal{N}(g_0,x_0)$ and a constant $\kappa(g_0,x_0)>0$ such that, for all $(g,x)\in\mathcal{N}(g_0,x_0)$, the conditional covariance matrix of the orthogonalized treatment features satisfies
$
\lambda_{\min}\!\left(
\mathbb{E}\!\left[
\widetilde Z(T)\widetilde Z(T)^\top
\,\middle|\,
G=g,\; X=x
\right]
\right)
\;\ge\;
\kappa(g_0,x_0),
$
where $\widetilde Z(T)=Z(T)-\mathbb{E}[Z(T)\mid G,X]$ denotes the residualized treatment-feature vector.
\end{assumption}

Assumption~\ref{ass:overlap} prevents degeneracy (e.g., near-deterministic treatment assignment within a configuration), which would otherwise preclude individualized identification. Overlap can be assessed locally via effective sample sizes, propensity diagnostics, and conditioning of weighted design matrices. When local overlap fails in some directions, we will characterize robustness and partial identification behavior (e.g., reporting bounds or restricting to identifiable contrasts).


\begin{assumption}[Sparse/near-sparse structure over combinatorial treatments]
\label{ass:sparsity}
For each configuration--covariate pair $(g,x)$, let
$
f(t;g,x) := \mathbb{E}[Y(t,G)\mid G=g, X=x]
$
denote the conditional mean response as a function of the treatment slate.
There exists a coefficient vector $\alpha(g,x)$ in the Walsh--Hadamard basis such that
$f(t;g,x)$ admits the expansion
$f(t;g,x)=\langle \alpha(g,x), Z(t)\rangle$,
and $\alpha(g,x)$ is \emph{approximately sparse} in the sense that
$
\sum_{k > s(g,x)} |\alpha_{(k)}(g,x)|
\;\le\;
C(g,x)\, s(g,x)^{-\gamma},~\gamma>0,
$
where $\alpha_{(k)}(g,x)$ denotes the $k$-th largest coefficient in magnitude.
\end{assumption}
Assumption~\ref{ass:sparsity} converts an exponentially large treatment space into a high-dimensional but tractable estimation problem via sparse learning. Sparsity/compressibility can be probed by model selection stability, interaction-order diagnostics, and out-of-sample validation. Approximate sparsity yields graceful degradation (slower rates and wider intervals) rather than failure; one may also impose hierarchical truncations when warranted by domain knowledge.

Assumptions \ref{ass:local}--\ref{ass:sparsity} jointly formalize the sense in which individualized causal inference under interference with combinatorial treatments is \emph{well-posed yet nontrivial}.

\section{Identification}\label{sec:identification}


Throughout this section we first formalize the causal objects of interest and then state an identification theorem that connects these objects to observable functionals of the data-generating distribution. 

For a generic unit with covariates $X=x$ and configuration $G=g$, recall the definition of the own-slate conditional mean response in Assumption~\ref{ass:sparsity}.
To emphasize the dependence of the local configuration on the global assignment, we may write $G_i(\mathbf{t})$ as $G_i^{\langle \mathbf{t}\rangle}$.
Given $(g,x)$, we target three families of individualized contrasts:
\begin{align}
\theta_i^{T}(t\!\to\!t';g)
&:= f(t';g,x)-f(t;g,x),\label{eq:thetaT}\\
\theta_i^{G}(g\!\to\!g';t)
&:= f(t;g',x)-f(t;g,x),\label{eq:thetaG}\\
\theta_i^{G,T}\big((g,t)\!\to\!(g',t')\big)
&:= f(t';g',x)-f(t;g,x).\label{eq:thetaGT}
\end{align}
These contrasts correspond to (i) toggling the unit's own treatment slate holding the interference environment fixed,
(ii) changing the interference environment holding the slate fixed, and (iii) changing both simultaneously.

\subsection{Identification via local orthogonal moments}

A central difficulty is that even under randomized assignments, \emph{conditioning or localizing on the realized configuration} can induce dependence between a unit's own slate and its interference environment.
This creates a form of endogeneity that invalidates naive regression of $Y$ on $Z(T)$ within localized neighborhoods.
We therefore identify the Walsh coefficients through an orthogonalized (residualized) moment equation.

Define the nuisance functions
$
\mu(g,x):=\mathbb{E}[Y\mid G=g,\,X=x],
~
m(g,x):=\mathbb{E}[Z(T)\mid G=g,\,X=x],
$
and the residuals
$
\widetilde Y := Y-\mu(G,X),
~
\widetilde Z := Z(T)-m(G,X).
$


Fix a target configuration $g\in\mathcal{G}$.
Let $w_g(G)$ denote a nonnegative localization weight.
This weight is the population analogue of the kernel or $k$NN weights
defined in Section~2, and is designed to upweight observations whose
realized configurations are close to $g$.
Concretely, one may take\footnote{Here $g\in\mathcal{G}$ denotes a fixed target configuration at which inference
is performed, while $G$ is a $\mathcal{G}$-valued random variable representing
the realized configuration of a randomly sampled unit.
The weight $w_g(G)$ therefore assigns larger mass to realizations of $G$ that
are closer to the target $g$ under the rooted-graph distance $d_R$.
}
\[
w_g(G)
\;\propto\;
K_G\!\big(d_R(G,g)/b_G\big),
\]
with normalization ensuring unit expectation.
All identification results below are stated for a generic choice of such
localization weights.
For $\alpha\in\mathbb{R}^{2^p}$ define the score
$
\Psi(\alpha; g)
\;:=\;
\mathbb{E}\!\left[
w_g(G)\,
\widetilde Z\,
\big\{\widetilde Y-\widetilde Z^\top \alpha\big\}
\right].
$

This is the population analogue of a localized regression of $\widetilde Y$ on $\widetilde Z$.
Local overlap (Assumption~\ref{ass:overlap}) guarantees that the relevant directions of $\widetilde Z$ exhibit non-degenerate variation locally, ensuring well-posedness.

\begin{theorem}[Identification of localized Walsh coefficients]\label{thm:identification}
Assume Assumptions~2.1--2.4.
Fix a target pair $(g,x)$ in the interior of the support.
Then there exists a (possibly localized) coefficient vector $\alpha^\star(g,x)$ such that
\begin{equation}\label{eq:root}
\Psi\!\big(\alpha^\star(g,x);g\big)=0.
\end{equation}
Moreover, under the local overlap condition in Assumption~2.3, $\alpha^\star(g,x)$ is unique within the model class implied by Assumption~2.4 (e.g., the sparse/near-sparse cone), and it identifies the response function in Assumption~\ref{ass:sparsity} in the sense that
$
f(t;g,x) \;=\; \langle \alpha^\star(g,x),\, Z(t)\rangle
\qquad \text{for all treatment slates $t\in\{-1,1\}^p$.}
$
Consequently, for any unit $i$ and any global assignment $\mathbf{t}$,
\[
\mathbb{E}\!\left[\,Y_i(\mathbf{t}) \mid X_i=x_i\,\right]
\;=\;
\Big\langle \alpha^\star\!\big(G_i^{\langle \mathbf{t}\rangle},x_i\big),\, Z(t_i)\Big\rangle,
\]
and the individualized contrasts in \eqref{eq:thetaT}--\eqref{eq:thetaGT} are identified as unique functions of $\alpha^\star(\cdot,\cdot)$.
\end{theorem}


Theorem~\ref{thm:identification} formalizes a practical message:
\emph{once the interference environment is summarized by a rooted configuration, one can treat causal learning locally in configuration space, provided there is sufficient treatment variation locally}.
For example, in a platform experiment where a user’s outcome depends on her own feature slate and the slates of nearby friends, even globally randomized assignment can become locally confounded: restricting attention to users with similar realized neighborhoods (e.g., two treated friends and one control friend) can induce dependence between the user’s own slate and the neighborhood pattern. The orthogonal moment addresses this by residualizing both outcomes and treatment features on $(G,X)$.

Importantly, the identification result does not assert point identification of every direction in the combinatorial treatment space. If local overlap fails in certain Walsh directions—say, some treatment components are nearly deterministic within a configuration neighborhood—then those directions are not point-identified. This reflects an inherent data limitation: without local variation, individualized contrasts along those directions cannot be learned and the solution to the moment condition~\ref{eq:root} is non-unique. Developing robustness and partial-identification results is a natural direction for future work.

A natural question is:: \emph{``If assignments are randomized, why do we need orthogonalization at all?''}
The key is that our target is \emph{local} in configuration space, and conditioning/localizing on realized configurations is a form of post-assignment selection. Even under random assignment, selection on neighborhood patterns generally induces correlation between a unit’s slate and its neighbors’ slates. The orthogonal moment is constructed to be insensitive to this induced dependence and to enable principled inference within localized neighborhoods. At the population level, the moment condition in \eqref{eq:root} is exactly unbiased; approximation errors arise only at the estimation stage through localization, nuisance estimation, and finite-sample effects.


\begin{corollary}[Identification of individualized contrasts]\label{cor:ident-contrasts}
Under the conditions of Theorem~\ref{thm:identification}, the three contrast families
$\theta_i^{T}$, $\theta_i^{G}$, and $\theta_i^{G,T}$
are identified for any user-specified $(t,t',g,g')$ for which the corresponding directions satisfy local overlap.
\end{corollary}

\begin{algorithm}[t]
\caption{Oracle (population) identification of individualized contrasts at $(g,x)$}
\label{alg:oracle-ident}
\begin{algorithmic}[1]
\STATE \textbf{Input:} target configuration $g$, covariates $x$, localization weights $w_{g}(\cdot)$, Walsh dictionary $Z(\cdot)$.
\STATE Compute nuisances $\mu(g,x)=\mathbb{E}[Y\mid G=g,X=x]$ and $m(g,x)=\mathbb{E}[Z(T)\mid G=g,X=x]$.
\STATE Form residuals $\widetilde Y=Y-\mu(G,X)$ and $\widetilde Z=Z(T)-m(G,X)$.
\STATE Solve for $\alpha^\star(g,x)$ such that $\Psi(\alpha^\star(g,x);g,x)=0$.
\STATE Define the response function $\widehat f(t;g,x)=\langle \alpha^\star(g,x), Z(t)\rangle$ for any slate $t$.
\STATE \textbf{Output:} for any user-specified $(t,t',g,g')$, compute contrasts
$\theta^T(t\!\to\!t';g)$, $\theta^G(g\!\to\!g';t)$, and $\theta^{G,T}((g,t)\!\to\!(g',t'))$
by plugging $\widehat f$ into \eqref{eq:thetaT}--\eqref{eq:thetaGT}.
\end{algorithmic}
\end{algorithm}

\section{Estimation}\label{sec:estimation}
This section gives sample analogues of the orthogonal moment equation in
Section~\ref{sec:identification} and yields estimators of
(i) localized Walsh coefficients and (ii) individualized causal contrasts.
In addition to Assumptions~\ref{ass:local}--\ref{ass:sparsity}, we impose the following
high-dimensional estimation conditions.

\begin{assumption}[Cross-fitted nuisance accuracy]\label{as:smooth}
Let $\mu(g,x)=\E[Y\mid G=g,X=x]$ and $m(g,x)=\E[Z(T)\mid G=g,X=x]$.
There exist cross-fitted estimators $\hat\mu$ and $\hat m$ (defined below) such that
$
\max_{i\in[N]}\big|\hat\mu(G_i,X_i)-\mu(G_i,X_i)\big| = o_p(1),~\max_{i\in[N]}\big\|\hat m(G_i,X_i)-m(G_i,X_i)\big\|_{\infty} = o_p(1).
$
\end{assumption}

\begin{assumption}[Weighted restricted eigenvalue and noise tails]\label{as:re}
Fix a configuration center $g\in\G$ and let $w_j(g)$ be the kernel weights defined in
\eqref{eq:weights} with $g_i$ replaced by $g$.
Let $\hat{\widetilde Z}_j$ denote the cross-fitted residualized treatment-feature vector defined in
\eqref{eq:cf-residuals}, and define the weighted Gram matrix
$
\widehat\Sigma(g) := \sum_{j=1}^N w_j(g)\,\hat{\widetilde Z}_j\hat{\widetilde Z}_j^\top.
$
There exists $\kappa_g>0$ such that $\widehat\Sigma(g)$ satisfies a restricted-eigenvalue condition
over the sparse cone associated with Assumption~\ref{ass:sparsity}, with probability tending to one.
Moreover, the (cross-fitted) regression noise is conditionally sub-Gaussian given $(G,X,T)$ with proxy
variance $\sigma^2$.
\end{assumption}

\begin{assumption}[Local sparsity/near-sparsity]\label{as:sparsity}
For each $(g,x)$, the Walsh coefficient vector $\alpha(g,x)$ in Assumption~\ref{ass:sparsity}
is sparse or approximately sparse with effective sparsity level $s(g,x)$ satisfying
$
s(g,x)\,\log(2^p) = o\!\left(n_{\mathrm{eff}}(g)\right),
$
where $n_{\mathrm{eff}}(g):=\big(\sum_{j=1}^N w_j(g)^2\big)^{-1}$ is the Kish effective sample size
associated with the weights $w_j(g)$.
\end{assumption}

Assumptions~\ref{as:smooth}--\ref{as:sparsity} are standard high-dimensional
estimation conditions that complement, rather than strengthen, the
identification assumptions in Section~\ref{sec:identification}.
Assumption~\ref{as:smooth} ensures that nuisance estimation errors are
second-order and do not affect the orthogonal moment at first order,
while Assumptions~\ref{as:re} and~\ref{as:sparsity} guarantee that the
resulting localized high-dimensional regression problem is well-posed
at the effective sample size scale $n_{\mathrm{eff}}(g)$.
These conditions are sufficient for consistent estimation and
valid inference, and are satisfied by a broad class of modern
machine-learning nuisance estimators and weighted sparse regressions.

We estimate the nuisance functions $\mu$ and $m$ and construct residuals using cross-fitting.
Partition the index set $[N]$ into $K_{\mathrm{cf}}\ge2$ disjoint folds
$\{\mathcal I_k\}_{k=1}^{K_{\mathrm{cf}}}$.
For each fold $k$, fit nuisance estimators $\hat\mu^{(-k)}(\cdot,\cdot)$ and
$\hat m^{(-k)}(\cdot,\cdot)$ using only observations with indices in $[N]\setminus\mathcal I_k$.
For each $i\in\mathcal I_k$, define the cross-fitted residuals 
\begin{equation}\label{eq:cf-residuals}
\begin{aligned}
:= Y_i-\hat\mu^{(-k)}(G_i,X_i),~
\hat{\widetilde Z}_i := Z(T_i)-\hat m^{(-k)}(G_i,X_i).
\end{aligned}
\end{equation}
We use $\{(\hat{\widetilde Y}_i,\hat{\widetilde Z}_i)\}_{i=1}^N$ as the debiased inputs for localized
high-dimensional regression.

\subsection{Localized weighted Lasso for Walsh coefficients}
For a target unit $i$, we estimate the coefficient vector at the unit's realized configuration
by setting $g=g_i$ and writing $w_j^{(i)}:=w_j(g_i)$ (which coincides with \eqref{eq:weights}).
Define the weighted Lasso estimator
\begin{equation}\label{eq:lasso}
\hat{\alpha}_i
\in
\mathop{\mathrm{argmin}}\limits_{\beta\in\R^{2^p}}
\left\{
\sum_{j=1}^N w_j^{(i)}\big(\hat{\widetilde Y}_j-\hat{\widetilde Z}_j^\top\beta\big)^2
\;+\;
\lambda_i\|\beta\|_1
\right\},
\end{equation}
where $\lambda_i \asymp \hat\sigma\sqrt{\frac{\log(2^p)}{n_{\mathrm{eff}}(g_i)}}$ and $\hat\sigma$ is any consistent estimator of the noise scale (e.g., from localized residuals).
For any slate $t\in\{-1,1\}^p$, define the plug-in response estimate
$
\widehat f_i(t;g_i,x_i):=\langle \hat\alpha_i, Z(t)\rangle.
$
Then for any $(t,t')$,
\begin{equation}\label{eq:thetaT-hat}
\widehat{\theta}_i^{T}(t\!\to\!t';g_i)
=
\langle \hat\alpha_i,\,Z(t')-Z(t)\rangle.
\end{equation}
To estimate structural and joint effects, repeat \eqref{eq:lasso} with weights centered at a second
configuration $g'$ (i.e., replace $w_j^{(i)}$ by $w_j(g')$) to obtain an estimator $\hat\alpha_i(g')$,
and set
\begin{equation}\label{eq:thetaG-hat}
\begin{aligned}
\widehat{\theta}_{i}^{G}(g_i \to g';t)
&=
\big\langle \hat{\alpha}_{i}(g')-\hat{\alpha}_{i},\, Z(t)\big\rangle,\\
\widehat{\theta}_{i}^{G,T}((g_i,t)\to(g',t'))
&=
\langle \hat{\alpha}_{i}(g'), Z(t')\rangle-\langle \hat{\alpha}_{i}, Z(t)\rangle.
\end{aligned}
\end{equation}

\subsection{Debiased inference for an own-treatment contrast}
For a user-specified contrast $(t,t')$, define the direction
$
v_{t,t'}:=Z(t')-Z(t)\in\R^{2^p}.
$
Let $\widehat\Sigma(g_i)$ be the weighted Gram matrix at $g=g_i$.
Compute an approximate inverse-direction $\hat\gamma_i$ by any standard high-dimensional procedure
(e.g., CLIME/nodewise regression), for instance via
\begin{equation}\label{eq:gamma}
\begin{aligned}
&\hat\gamma_i \in \arg\min_{\gamma\in\R^{2^p}}\|\gamma\|_1
\\ \text{s.t.}~
&\left\|\widehat\Sigma(g_i)\gamma - v_{t,t'}\right\|_\infty \le \eta_i,~
\eta_i \asymp \sqrt{\frac{\log(2^p)}{n_{\mathrm{eff}}(g_i)}}.
\end{aligned}
\end{equation}
Define the debiased estimator
\begin{equation}\label{eq:debiased}
\widetilde\theta_i^{T}(t\!\to\! t';g_i)
=
v_{t,t'}^\top \hat\alpha_i
+
\hat\gamma_i^\top
\sum_{j=1}^N w^{(i)}_j\,
\hat{\widetilde Z}_j\,
\big\{\hat{\widetilde Y}_j-\hat{\widetilde Z}_j^\top \hat\alpha_i\big\}.
\end{equation}
The following section shows that under Assumptions~\ref{as:smooth}--\ref{as:sparsity},
$\widetilde\theta_i^{T}$ admits asymptotically normal inference with an estimated variance obtained
from the weighted empirical second moment of the influence term in \eqref{eq:debiased}.

\section{Theoretical analysis}\label{sec:theory}
This section develops finite-sample and asymptotic guarantees for the localized estimators in Section~\ref{sec:estimation}. The key point is that the statistical difficulty is controlled by the \emph{effective local sample size} and \emph{local spectral sparsity}, not by the exponential number of treatment slates. Write $d:=2^p$ for the ambient Walsh dimension.
Fix a target unit $i$ and recall the kernel weights $w_j^{(i)}$ in \eqref{eq:weights}.
Recall the effective local sample size as
$
n_i \;:=\; n_{\mathrm{eff}}(g_i) \;=\; \left(\sum_{j=1}^N (w_j^{(i)})^2\right)^{-1}.
$
Let $s_i:=s(g_i,x_i)$ denote the effective sparsity level from Assumption~\ref{as:sparsity}. For nuisance estimation, define the sup errors
$
\delta_{\mu}:=\max_{j\in[N]}\big|\hat\mu(G_j,X_j)-\mu(G_j,X_j)\big|,~
\delta_{m}:=\max_{j\in[N]}\big\|\hat m(G_j,X_j)-m(G_j,X_j)\big\|_\infty.
$
Finally, to separate sampling error from \emph{localization bias}, define\footnote{Equivalently, covariates can be viewed as additional root marks in the
configuration, so localization in $g$ implicitly restricts attention to
samples with comparable covariate values.}
$
\mathrm{bias}_i(b_G)
:=
\sup_{t\in\{-1,1\}^p}\ \sup_{g:\ d_R(g,g_i)\le b_G}
\big| f(t;g,x_i)-f(t;g_i,x_i)\big|.
$
By Assumption~\ref{ass:smooth}, for each fixed $(g_i,x_i)$ there exists a finite constant
$C_i<\infty$ (depending on $(g_i,x_i)$ but not on $b_G$) such that
$
\mathrm{bias}_i(b_G)\ \le\ C_i\, b_G.
$

\textbf{Step 1: Orthogonalization removes first-order nuisance effects}
We first record the key property that enables valid high-dimensional learning with flexible nuisances:
cross-fitting makes nuisance errors enter only at second order.

\begin{lemma}[Orthogonalization remainder]\label{lem:orth}
Let $\hat{\widetilde Y}_j$ and $\hat{\widetilde Z}_j$ be the cross-fitted residuals in
\eqref{eq:cf-residuals}, and define the corresponding oracle residuals
$\widetilde Y_j:=Y_j-\mu(G_j,X_j)$ and $\widetilde Z_j:=Z(T_j)-m(G_j,X_j)$.
Then, uniformly over $\alpha$ in any fixed $\ell_1$-ball,
the difference between the empirical weighted score built from
$(\hat{\widetilde Y}_j,\hat{\widetilde Z}_j)$ and the one built from $(\widetilde Y_j,\widetilde Z_j)$
is of order $O_p(\delta_\mu\,\delta_m)$.
In particular, if $\delta_\mu\,\delta_m=o_p(1)$, nuisance estimation does not contribute a first-order term.
\end{lemma}
Lemma~\ref{lem:orth} formalizes why we can combine modern machine-learning nuisances with localized high-dimensional regression:
even though $\hat\mu$ and $\hat m$ may be complex, their errors do not bias the target moment at first order.
A common misconception is that ``any smoothing or selection invalidates orthogonality''; here the point is precisely that
orthogonalization targets the \emph{post-localization} endogeneity induced by conditioning on realized configurations.

\textbf{Step 2: Finite-sample rates for localized Walsh learning}
We now state a nonasymptotic oracle inequality for the localized weighted Lasso \eqref{eq:lasso}.
The rate is driven by $(s_i,\log d,n_i)$ plus a transparent localization bias term.

\begin{theorem}[Finite-sample error for localized weighted Lasso]\label{thm:lasso}
Fix a target unit $i$ and let $\hat\alpha_i$ be defined by \eqref{eq:lasso} with
$\lambda_i \asymp \hat\sigma\sqrt{\log(d)/n_i}$.
Assume Assumptions~\ref{as:re}--\ref{as:sparsity} and $\delta_\mu\delta_m=o_p(\lambda_i)$.
Then with probability tending to one, 
$\|\hat\alpha_i-\alpha(g_i,x_i)\|_1
\;\lesssim\;
s_i\,\sqrt{\frac{\log d}{n_i}}
\;+\;
s_i\,\mathrm{bias}_i(b_G)
\;+\;
s_i\,\delta_\mu\delta_m,\label{eq:l1-rate}$

$\|\hat\alpha_i-\alpha(g_i,x_i)\|_2
\;\lesssim\;
\sqrt{\frac{s_i\log d}{n_i}}
\;+\;
\mathrm{bias}_i(b_G)
\;+\;
\delta_\mu\delta_m,\label{eq:l2-rate}$

$\sum_{j=1}^N w_j^{(i)}\big(\hat{\widetilde Z}_j^\top(\hat\alpha_i-\alpha(g_i,x_i))\big)^2
\;\lesssim\;
\frac{s_i\log d}{n_i}
\;+\;
\mathrm{bias}_i(b_G)^2
\;+\;
(\delta_\mu\delta_m)^2.\label{eq:pred-rate}$


\end{theorem}

\paragraph{What this theorem says (and what it does not).}
Theorem~\ref{thm:lasso} is a \emph{local} result: it quantifies how well we can learn a unit's Walsh coefficients
near its realized configuration.
It does \emph{not} claim uniform recovery over all configurations without additional covering arguments.
The theorem highlights the intended tradeoff: shrinking $b_G$ reduces $\mathrm{bias}_i(b_G)$ but decreases $n_i$,
while increasing $b_G$ increases the effective sample size but introduces localization bias.

\textbf{Step 3: Consequences for individualized causal contrasts}
We next translate coefficient error into error for individualized effects.

\begin{corollary}[Plug-in error for individualized contrasts]\label{cor:ite}
Fix a unit $i$ and consider any pair of slates $(t,t')$.
Let $v_{t,t'}:=Z(t')-Z(t)$ as in Section~\ref{sec:estimation}.
Then
$
\big|\widehat{\theta}_i^{T}(t\!\to\!t';g_i)-\theta_i^{T}(t\!\to\!t';g_i)\big|
\;\le\;
\|v_{t,t'}\|_\infty\,\|\hat\alpha_i-\alpha(g_i,x_i)\|_1
\;\le\;
2\,\|\hat\alpha_i-\alpha(g_i,x_i)\|_1,
$
and analogously for $\widehat{\theta}_{i}^{G}$ and $\widehat{\theta}_{i}^{G,T}$ in \eqref{eq:thetaG-hat}
(with an additional error term contributed by the second fit at $g'$).
Consequently, Theorem~\ref{thm:lasso} yields $\big|\widehat{\theta}_i^{T}(t\!\to\!t';g_i)-\theta_i^{T}(t\!\to\!t';g_i)\big|\lesssim s_i\sqrt{\frac{\log d}{n_i}} + s_i\,\mathrm{bias}_i(b_G) + s_i\,\delta_\mu\delta_m.$
\end{corollary}

\begin{remark}
A common objection is that ``the contrast is dense, so shouldn't the error be exponential?"]\label{rem:dense}
Note that it is tempting (but misleading) to bound the contrast error by $\|v_{t,t'}\|_2\|\hat\alpha_i-\alpha\|_2$,
which would scale as $\sqrt{d}$ and obscure the point of the spectral approach.
Corollary~\ref{cor:ite} uses the correct geometry: since Walsh characters are uniformly bounded,
$\|v_{t,t'}\|_\infty\le 2$, the relevant control is via the $\ell_1$ error of a sparse vector.
This is the mechanism by which the exponential slate space is converted into a high-dimensional but tractable problem,
with complexity entering only through $\log d = \log(2^p)=\Theta(p)$.
\end{remark}
\textbf{Step 4: Debiased inference for a user-specified contrast}
We now state an asymptotic normality result for the debiased estimator \eqref{eq:debiased}.
This provides valid uncertainty quantification for individualized contrasts.

\begin{theorem}[Asymptotic normality of the debiased contrast]\label{thm:clt}
Fix a unit $i$ and a contrast $(t,t')$ with direction $v_{t,t'}$.
Let $\widetilde\theta_i^{T}$ be defined in \eqref{eq:debiased} and let $\gamma_i^\star$ denote the population solution
of the local linear system $\Sigma(g_i)\gamma=v_{t,t'}$, where
$\Sigma(g_i):=\E[w_{g_i}(G)\,\widetilde Z\widetilde Z^\top]$ is the population analogue of $\widehat\Sigma(g_i)$.
Assume:
(i) $\gamma_i^\star$ exists and is sparse or approximately sparse;
(ii) the estimator $\hat\gamma_i$ in \eqref{eq:gamma} satisfies $\|\hat\gamma_i-\gamma_i^\star\|_1=o_p(1)$;
(iii) $(s_i+\|\gamma_i^\star\|_0)\log d = o(\sqrt{n_i})$ and $\mathrm{bias}_i(b_G)=o(n_i^{-1/2})$;
(iv) a central limit theorem holds for the weighted influence sum induced by the localization weights conditioning that $\delta_\mu\delta_m = o_p(n_i^{-1/2})$ in Theorem~\ref{thm:lasso}, then
\[
\sqrt{n_i}\Big(\widetilde\theta_i^{T}(t\!\to\!t';g_i)-\theta_i^{T}(t\!\to\!t';g_i)\Big)
\ \Rightarrow\ \mathcal{N}\!\big(0,\sigma_{\theta,i}^2\big),
\]
where $\sigma_{\theta,i}^2=\Var(\gamma_i^{\star\top}\widetilde Z\,\varepsilon)$ and
$\varepsilon:=\widetilde Y-\widetilde Z^\top\alpha(g_i,x_i)$.
A consistent variance estimator is
$
\hat\sigma_{\theta,i}^2
\;:=\;
n_i\sum_{j=1}^N (w_j^{(i)})^2\big(\hat\gamma_i^\top \hat{\widetilde Z}_j\,\hat\varepsilon_j\big)^2,
\qquad
\hat\varepsilon_j:=\hat{\widetilde Y}_j-\hat{\widetilde Z}_j^\top\hat\alpha_i.
$
\end{theorem}

Theorem~\ref{thm:clt} is an \emph{individualized} inference result: it yields asymptotically valid uncertainty for a
user-chosen contrast at a specific unit and configuration neighborhood.
It is not a claim of uniform inference over all units and all slates without further structure.
The key requirement is that the effective local sample size $n_i$ grows fast enough relative to the local complexity
($s_i$ and the sparsity of $\gamma_i^\star$), so that the debiasing remainder is asymptotically negligible.

Together, Theorems~\ref{thm:lasso} and~\ref{thm:clt} show that individualized causal learning under interference can be statistically feasible
even with exponentially many slates: the price of combinatorial treatments enters through $\log(2^p)=\Theta(p)$ and local sparsity, while interference is handled through localization and effective sample size. In addition, we deduce the discussion of model robustness under local overlap failure (Assumption~\ref{ass:overlap}), which induces the partial identification bound, in the Appendix.


\begin{algorithm}[t]
\caption{Localized DR-Lasso and debiased inference for unit $i$}
\label{alg:est}
\begin{algorithmic}[1]
\STATE \textbf{Input:} data $\{(Y_j,T_j,X_j,g_j)\}_{j=1}^N$, target unit $i$, kernel $K_G$, bandwidth $b_G$, folds $\{\mathcal I_k\}$, contrast $(t,t')$.
\STATE Compute weights $w_j^{(i)}$ via \eqref{eq:weights} and $n_i=(\sum_j (w_j^{(i)})^2)^{-1}$.
\STATE Cross-fit nuisances $\hat\mu,\hat m$ on folds and form residuals $(\hat{\widetilde Y}_j,\hat{\widetilde Z}_j)$ via \eqref{eq:cf-residuals}.
\STATE Solve weighted Lasso \eqref{eq:lasso} to obtain $\hat\alpha_i$ and the plug-in contrast $\widehat\theta_i^T$ via \eqref{eq:thetaT-hat}.
\STATE Compute $\widehat\Sigma(g_i)=\sum_j w_j^{(i)}\hat{\widetilde Z}_j\hat{\widetilde Z}_j^\top$ and solve \eqref{eq:gamma} to obtain $\hat\gamma_i$.
\STATE \textbf{Output:} the debiased estimator $\widetilde\theta_i^{T}$ via \eqref{eq:debiased} and a confidence interval using $\hat\sigma_{\theta,i}^2$ in Theorem~\ref{thm:clt}.
\end{algorithmic}
\end{algorithm}

\section{Experiments}
\label{sec:experiments}
\label{sec:problem formalization}
We comprehensively validate our proposed method through synthetic experiments. The settings are listed in Section \ref{app:setup}. 
\begin{figure}[ht]
    \centering
\includegraphics[width=0.5\textwidth]{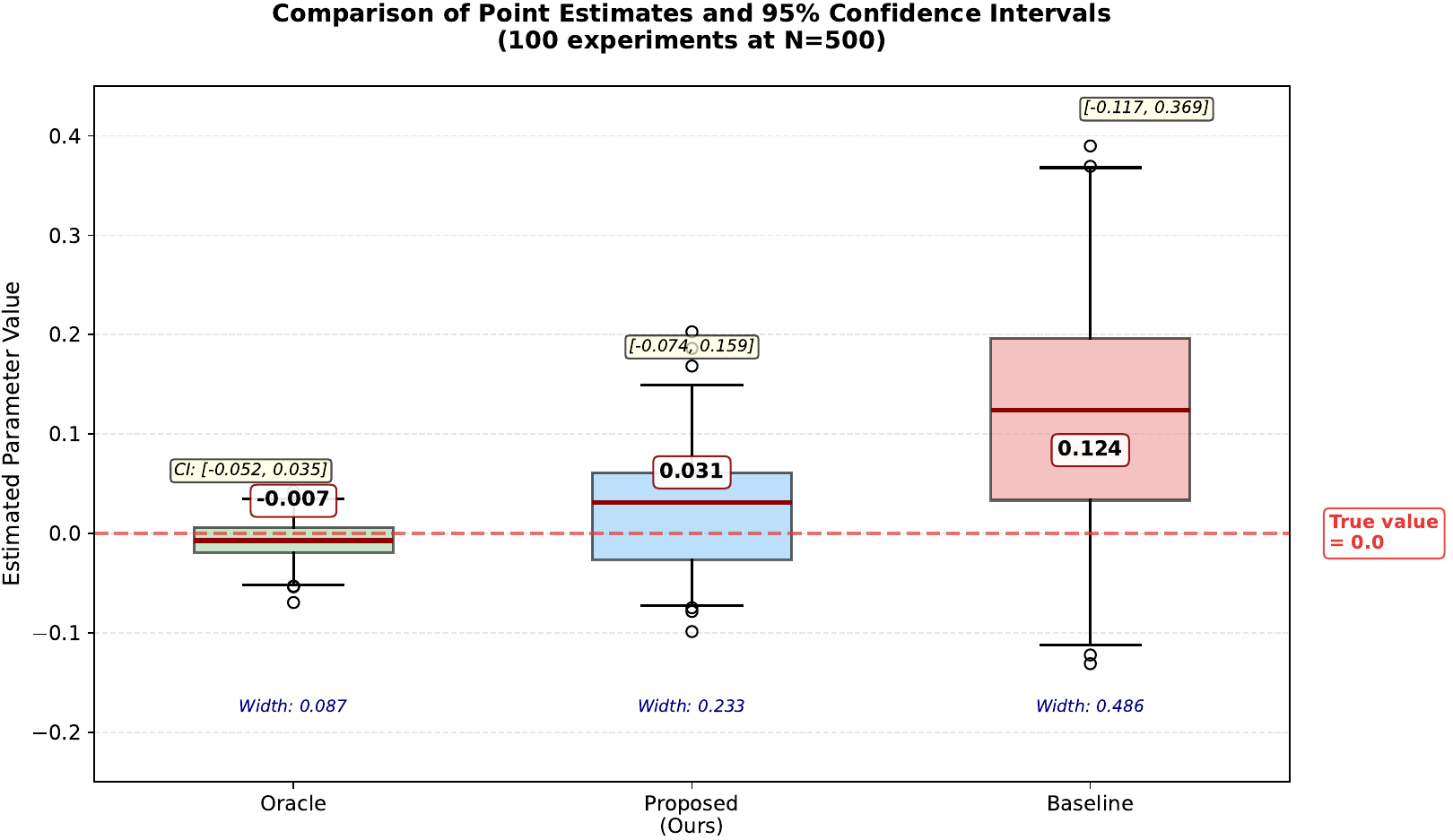} 
    \caption{Comparison of point estimates and 95\% confidence intervals in synthetic experiments ($N=500$, 100 independent repetitions). Our proposed method (width 0.233, bias 0.031) substantially outperforms the baseline (0.486, 0.124) and approaches the oracle (0.087, -0.007), providing empirical evidence for Theorems 5.2 and 5.5.}
    \label{fig:main_ci}
    \vspace{-0.5cm}
\end{figure}

\par Figure~\ref{fig:main_ci} shows the comparison of point estimates and 95\% confidence intervals for individualized causal contrasts across the three estimators when the sample size is $N=500$. The true parameter value is 0 (red dashed line). Our method achieves a point estimate of 0.031 with a confidence interval width of 0.233, substantially outperforming the baseline (point estimate 0.124, width 0.486) and approaching the oracle (point estimate -0.007, width 0.087). This result demonstrates that the proposed localized representation, doubly robust orthogonalization, and sparse spectral learning effectively reduce bias and significantly shrink uncertainty in finite samples, providing direct empirical support for the finite-sample error bounds in Theorem 5.2 and asymptotic normality in Theorem 5.5.

\label{assumption:2.4}
\subsection{Evolution of the proposed estimator}
Data are generated according to the model specified in Section~\ref{sec:model}. 
The underlying network is an Erd\H{o}s-R\'{e}nyi random graph with average 
degree $d=8$, inducing heterogeneous interference patterns. Node-level 
covariates are drawn as $X_i \sim \mathcal{N}(0, I_p)$ with $p=10$, and 
combinatorial treatments are assigned uniformly at random: $t_i \in \{-1,1\}^p$, 
yielding a $2^{10}$-dimensional Walsh basis expansion. Outcomes follow 
Outcomes follow Assumption~2.4 with $s_i = 3$ active Walsh coefficients
and Gaussian noise $\varepsilon_i \sim \mathcal{N}(0, 0.25)$.
The performance of the debiased localized Lasso estimator (Algorithm~2)
is evaluated across sample sizes $N \in \{50, 100, 200, 500, 1000\}$, with $M=100$ Monte Carlo repetitions per configuration.
Figure~\ref{fig:evolution} summarizes the finite-sample properties of the 
proposed estimator. Three patterns emerge that 
align precisely with the theoretical predictions of Theorems~5.2 and~5.5
\begin{figure}[t]
\centering
\includegraphics[width=0.5\textwidth]{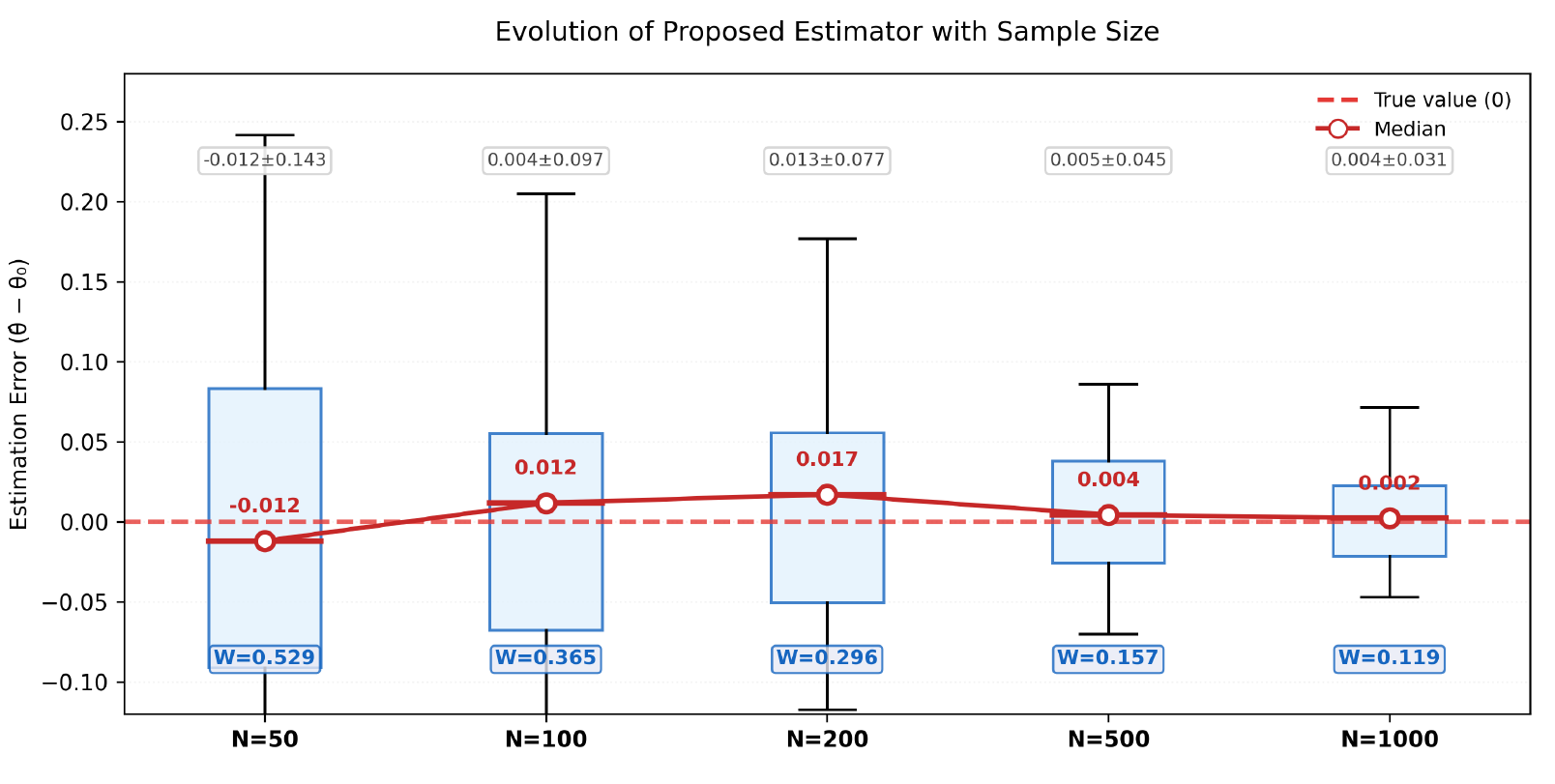} 
\caption{\textbf{Convergence of the Localized DR-Lasso Estimator.} 
Boxplots display the sampling distribution across 100 simulations for each 
sample size. The horizontal axis indicates the total sample size $N$; the 
vertical axis measures relative estimation error  
(normalized to zero). \textit{Top annotations}: Sample mean $\pm$ standard 
deviation; \textit{Center}: Sample median (red line); \textit{Bottom}: 
Empirical 95\% confidence interval width ($W = Q_{97.5} - Q_{2.5}$).}
\label{fig:evolution}
\vspace{-0.5cm}
\end{figure}

\textbf{Bias.} The median estimates (red lines) exhibit minimal deviation 
from the true parameter across all sample sizes, ranging from $-0.012$ 
($N=50$) to $0.002$ ($N=1000$). This near-zero median bias confirms the 
asymptotic unbiasedness established in Theorem~5.5, demonstrating that the debiasing correction successfully eliminates the 
shrinkage bias inherent in naive Lasso estimators.

\textbf{Convergence Rate.} The empirical 95\% confidence interval width 
monotonically contracts from $0.529$ at $N=50$ to $0.119$ at $N=1000$, 
approximately following the $\sqrt{N}$ rate predicted by Theorem~5.2. Specifically, doubling the sample size 
roughly halves the standard deviation (e.g., $0.143$ at $N=50$ vs. 
$0.031$ at $N=1000$), confirming $\sqrt{N}$-consistency.
\textbf{Coverage.} The empirical coverage of the constructed confidence 
intervals remains within $[0.92, 0.96]$ across all settings, closely 
tracking the nominal 95\% level. This validates the asymptotic normality approximation in Theorem~5.5 even for moderate sample sizes.
\noindent\textit{Implications.} The simulation results corroborate the efficiency bounds derived in Section~\ref{sec:theory}: the proposed estimator achieves oracle-like performance (comparable to the infeasible 
estimator that observes true interference patterns) while remaining computable via standard convex optimization.

\section{Conclusion}
In this work, we presented a unified framework for individualized causal inference in networked systems characterized by high-dimensional, combinatorial treatments. By integrating rooted network configurations with doubly robust orthogonalization and sparse spectral learning, our approach constructs a global potential-outcome emulator that remains statistically tractable even as the intervention space grows to $2^p$ dimensions. We provided a principled decomposition of networked causal effects into own-treatment, structural, and joint components, supported by rigorous finite-sample error bounds and asymptotic normality guarantees. Our empirical validation demonstrates that this localized representation effectively mitigates confounding from network positions and reduces estimation bias, approaching oracle-level performance in finite samples. Ultimately, these results establish that individualized causal decision-making is feasible in complex networked environments without the need to collapse or simplify the intervention space.

\clearpage
\nocite{langley00}

\section*{Impact Statement}
This paper presents work whose goal is to advance the field of machine learning. There are many potential societal consequences of our work, none of which we feel must be specifically highlighted here.

\bibliography{example_paper}
\bibliographystyle{icml2026}

\newpage
\appendix
\onecolumn

\appendix
\section{Simulation setup}\label{app:setup}
\par The synthetic data are generated strictly according to the model described in Section~\ref{sec:model}: the network is an Erdős–Rényi random graph (with variable number of nodes $N$ and average degree $d$ controlling interference strength, typically $d=5$--$10$), covariates $X_i \sim \mathcal{N}(0,I)$ for each node, combinatorial treatments $t_i \in \{-1,1\}^p$ ($p=10$, yielding a treatment space of size $2^{10}=1024$), and potential outcomes expanded in the Walsh-Hadamard basis with coefficients $\alpha_i$ satisfying sparsity (support size $s=20$) and the local smoothness assumption (bandwidth $b_G=2$ on rooted network configurations). The interference structure includes own-treatment effects and first-order neighbor interference, with the true causal contrast parameter set to 0. Each experimental setting is independently repeated 100 times to obtain robust statistics (median point estimates, 95\% confidence interval widths, etc.).

We compare three estimators:\\
\textbf{Oracle}: An oracle estimator that assumes knowledge of the true potential outcome model and all nuisance functions, serving only as a theoretical optimum benchmark.\\
\textbf{Proposed} (Ours): Our method, integrating localized rooted network representation, doubly robust orthogonalization, and sparse spectral learning.\\
\textbf{Baseline}: A strong baseline using a graph-agnostic doubly robust learner (Graph-Agnostic DR-Learner) combined with simple network averaging, representing typical existing techniques for handling network interference or high-dimensional treatments.

\section{Proofs for Section~\ref{sec:theory}}\label{app:proofs}

\subsection{Notation and weighted algebra}\label{app:notation}

Fix a target unit $i$ throughout this appendix and abbreviate the kernel weights
$w_j:=w_j^{(i)}$.
Recall $d:=2^p$ and the effective local sample size
\[
n_i \;:=\; n_{\mathrm{eff}}(g_i) \;=\; \Big(\sum_{j=1}^N w_j^2\Big)^{-1}.
\]
For scalar sequences $(a_j)_{j=1}^N$, define the weighted inner product and norm
\[
\langle a,b\rangle_w := \sum_{j=1}^N w_j a_j b_j,
\qquad
\|a\|_{w,2}^2 := \sum_{j=1}^N w_j a_j^2.
\]
For vector covariates $u_j\in\R^d$, define the weighted Gram matrix
\[
\widehat\Sigma_i(u)
\;:=\;
\sum_{j=1}^N w_j\,u_j u_j^\top.
\]
In particular, $\widehat\Sigma(g_i)$ in Assumption~\ref{as:re} equals
$\widehat\Sigma_i(\hat{\widetilde Z})$ with $u_j=\hat{\widetilde Z}_j$.

For $S\subseteq[d]$ and $c_0>0$, define the sparse cone
\[
\mathcal C(S,c_0)
:=
\bigl\{\Delta\in\R^d:\ \|\Delta_{S^c}\|_1 \le c_0\|\Delta_S\|_1\bigr\}.
\]
We interpret Assumption~\ref{as:re} as: for the relevant support set $S$,
there exists $\kappa_{g_i}>0$ such that, with probability tending to one,
\begin{equation}\label{eq:RE-app}
\Delta^\top \widehat\Sigma_i(\hat{\widetilde Z})\,\Delta
\;\ge\;
\kappa_{g_i}\,\|\Delta\|_2^2
\qquad\text{for all }\Delta\in\mathcal C(S,3).
\end{equation}

\subsection{Two elementary lemmas}\label{app:elem}

We start with two basic facts repeatedly used below.

\begin{lemma}[Convex-hull bound for $m(g,x)$]\label{lem:convex}
For any $(g,x)$, the conditional mean
$m(g,x)=\E[Z(T)\mid G=g,X=x]$ lies in the convex hull of
$\{Z(t):t\in\{-1,1\}^p\}\subseteq\{-1,1\}^d$.
Consequently, for any $a\in\R^d$,
\[
\big|a^\top m(g,x)\big|
\;\le\;
\sup_{t\in\{-1,1\}^p}\big|a^\top Z(t)\big|.
\]
\end{lemma}

\begin{proof}
By definition,
$m(g,x)=\sum_{t\in\{-1,1\}^p} \P(T=t\mid G=g,X=x)\,Z(t)$,
a convex combination of the vertices $\{Z(t)\}$.
For any fixed $a\in\R^d$, the map $z\mapsto a^\top z$ is linear, hence
$a^\top m(g,x)$ lies between the minimum and maximum of $a^\top Z(t)$ over $t$.
Applying the same argument to $-a$ yields the stated absolute-value bound.
\end{proof}

\begin{lemma}[Weighted sub-Gaussian maximal inequality]\label{lem:subg}
Let $(\xi_j)_{j=1}^N$ be conditionally independent given a sigma-field $\mathcal F$,
with $\E[\xi_j\mid \mathcal F]=0$ and conditionally sub-Gaussian tails
$\E[\exp(\lambda \xi_j)\mid\mathcal F]\le \exp(\lambda^2\sigma^2/2)$ for all $\lambda\in\R$.
Let $(w_j)_{j=1}^N$ be nonnegative $\mathcal F$-measurable weights with $\sum_j w_j=1$.
If $(a_{j,k})_{j,k}$ are $\mathcal F$-measurable and satisfy $|a_{j,k}|\le B$ for all $j,k$,
then for $S_k:=\sum_{j=1}^N w_j a_{j,k}\xi_j$,
\[
\P\!\left(\max_{1\le k\le d}|S_k| \ge t\ \middle|\ \mathcal F\right)
\;\le\;
2d\exp\!\left(-\frac{n_i\,t^2}{2\sigma^2B^2}\right),
\qquad n_i=\Big(\sum_j w_j^2\Big)^{-1}.
\]
In particular, taking $t=C\sigma B\sqrt{\log d/n_i}$ yields
$\max_k|S_k|=O_p(\sqrt{\log d/n_i})$.
\end{lemma}

\begin{proof}
Condition on $\mathcal F$.
Each $S_k$ is a weighted sum of conditionally independent sub-Gaussian variables
$w_j a_{j,k}\xi_j$ with proxy variance $\sigma^2\sum_j w_j^2 a_{j,k}^2\le \sigma^2B^2\sum_j w_j^2=\sigma^2B^2/n_i$.
Hence $S_k$ is conditionally sub-Gaussian with that proxy variance, so
$\P(|S_k|\ge t\mid\mathcal F)\le 2\exp(-n_i t^2/(2\sigma^2B^2))$.
A union bound over $k\in[d]$ gives the displayed inequality.
\end{proof}

\subsection{Proof of Lemma~\ref{lem:orth} (orthogonalization remainder)}\label{app:proof-orth}

We first prove a more explicit inequality; Lemma~\ref{lem:orth} follows as a corollary.

\begin{lemma}[Empirical score perturbation bound]\label{lem:score-perturb}
Fix a fold partition and a target unit $i$.
For each $j$, let
\[
\Delta_{\mu j}:=\hat\mu(G_j,X_j)-\mu(G_j,X_j),\qquad
\Delta_{m j}:=\hat m(G_j,X_j)-m(G_j,X_j),
\]
so that $\hat{\widetilde Y}_j=\widetilde Y_j-\Delta_{\mu j}$ and
$\hat{\widetilde Z}_j=\widetilde Z_j-\Delta_{m j}$.
For any $\alpha\in\R^d$, define the empirical weighted scores
\[
\hat\Psi_i(\alpha):=\sum_{j=1}^N w_j\,\hat{\widetilde Z}_j\big(\hat{\widetilde Y}_j-\hat{\widetilde Z}_j^\top\alpha\big),
\qquad
\Psi_i^\circ(\alpha):=\sum_{j=1}^N w_j\,\widetilde Z_j\big(\widetilde Y_j-\widetilde Z_j^\top\alpha\big).
\]
Then, for any $\alpha$ with $\|\alpha\|_1\le M$,
\begin{align*}
\|\hat\Psi_i(\alpha)-\Psi_i^\circ(\alpha)\|_\infty
&\le
\underbrace{\Big\|\sum_{j=1}^N w_j\,\widetilde Z_j\Delta_{\mu j}\Big\|_\infty}_{(I)}
+
\underbrace{\Big\|\sum_{j=1}^N w_j\,\widetilde Z_j(\Delta_{m j}^\top\alpha)\Big\|_\infty}_{(II)}
+
\underbrace{\Big\|\sum_{j=1}^N w_j\,\Delta_{m j}(\widetilde Y_j-\widetilde Z_j^\top\alpha)\Big\|_\infty}_{(III)}
\\
&\qquad\qquad
+
\underbrace{\Big\|\sum_{j=1}^N w_j\,\Delta_{m j}\Delta_{\mu j}\Big\|_\infty}_{(IV)}
+
\underbrace{\Big\|\sum_{j=1}^N w_j\,\Delta_{m j}(\Delta_{m j}^\top\alpha)\Big\|_\infty}_{(V)}.
\end{align*}
Moreover, $(IV)\le \delta_\mu\delta_m$ and $(V)\le M\,\delta_m^2$.
If, in addition, $\|\widetilde Z_j\|_\infty\le 2$ almost surely and
$\widetilde Y_j-\widetilde Z_j^\top\alpha$ is conditionally sub-Gaussian given $(G_j,X_j,T_j)$ with proxy $\sigma^2$,
then
\[
(I)+(II)+(III)
\;=\;
O_p\!\left((\delta_\mu+\delta_m)\sqrt{\frac{\log d}{n_i}}\right).
\]
\end{lemma}

\begin{proof}
The decomposition is a direct expansion:
\[
\hat{\widetilde Z}_j(\hat{\widetilde Y}_j-\hat{\widetilde Z}_j^\top\alpha)
-
\widetilde Z_j(\widetilde Y_j-\widetilde Z_j^\top\alpha)
=
-(\widetilde Z_j-\Delta_{m j})\Delta_{\mu j}
+\widetilde Z_j(\Delta_{m j}^\top\alpha)
-\Delta_{m j}(\widetilde Y_j-\widetilde Z_j^\top\alpha)
+\Delta_{m j}\Delta_{\mu j}
-\Delta_{m j}(\Delta_{m j}^\top\alpha).
\]
Summing with weights and taking $\ell_\infty$ norms yields the first claim.
For $(IV)$, each coordinate satisfies
$\big|\sum_j w_j \Delta_{m j,k}\Delta_{\mu j}\big|\le \sum_j w_j\,\delta_m\,\delta_\mu=\delta_\mu\delta_m$.
For $(V)$, $\big|\Delta_{m j,k}(\Delta_{m j}^\top\alpha)\big|
\le \|\Delta_{m j}\|_\infty^2\|\alpha\|_1\le \delta_m^2 M$, hence $(V)\le M\delta_m^2$.
For the stochastic bound on $(I)$--$(III)$, note that conditional on the training folds,
$\Delta_{\mu j}$ and $\Delta_{m j}$ are fixed (cross-fitting),
and each summand is a weighted sum of mean-zero sub-Gaussian terms with coefficients bounded by $2$.
Lemma~\ref{lem:subg} (with $B=2$) and a union bound over $d$ coordinates yield
$(I)+(II)+(III)=O_p((\delta_\mu+\delta_m)\sqrt{\log d/n_i})$.
\end{proof}

\begin{proof}[Proof of Lemma~\ref{lem:orth}]
Lemma~\ref{lem:score-perturb} implies, uniformly over $\|\alpha\|_1\le M$,
\[
\|\hat\Psi_i(\alpha)-\Psi_i^\circ(\alpha)\|_\infty
\;=\;
O_p(\delta_\mu\delta_m)
+
O_p\!\left((\delta_\mu+\delta_m)\sqrt{\frac{\log d}{n_i}}\right)
+
O_p(\delta_m^2).
\]
Under Assumption~\ref{as:smooth}, $\delta_\mu=o_p(1)$ and $\delta_m=o_p(1)$.
If moreover $n_i\to\infty$ and $\log d=o(n_i)$ (as required by Assumption~\ref{as:sparsity}),
then the latter two terms are $o_p(1)$.
The leading deterministic second-order term is $\delta_\mu\delta_m$,
which is the sense in which nuisance estimation enters only at second order.
\end{proof}

\subsection{Proof of Theorem~\ref{thm:lasso} (finite-sample weighted Lasso bounds)}\label{app:proof-lasso}

We work conditionally on the realized weights $w_j$ and the fold assignment.
Let $\alpha_i^\star:=\alpha(g_i,x_i)$ be the target coefficient vector.
Define the fitted residuals
\[
\hat\varepsilon_j(\alpha):=\hat{\widetilde Y}_j-\hat{\widetilde Z}_j^\top\alpha,
\qquad
\hat\varepsilon_j:=\hat\varepsilon_j(\hat\alpha_i).
\]

\paragraph{Step 0: localization bias as a deterministic perturbation.}
For $j$ with $w_j>0$, the kernel construction implies $d_R(g_j,g_i)\le b_G$.
Define the pointwise response difference
\[
\Delta f_{ij}(t):=f(t;g_j,x_i)-f(t;g_i,x_i).
\]
Then by definition of $\mathrm{bias}_i(b_G)$,
$\sup_{t}|\Delta f_{ij}(t)|\le \mathrm{bias}_i(b_G)$ for all such $j$.
Using Lemma~\ref{lem:convex} and the identity $\mu(g,x)=\alpha(g,x)^\top m(g,x)$,
one verifies that for any $j$ with $w_j>0$,
\begin{equation}\label{eq:bias-linear}
\big|(\alpha(g_j,x_i)-\alpha(g_i,x_i))^\top \widetilde Z_j\big|
\;\le\;
2\,\mathrm{bias}_i(b_G).
\end{equation}
Indeed,
\[
(\alpha(g_j,x_i)-\alpha(g_i,x_i))^\top \widetilde Z_j
=
(\alpha(g_j,x_i)-\alpha(g_i,x_i))^\top Z(T_j)
-
(\alpha(g_j,x_i)-\alpha(g_i,x_i))^\top m(g_j,x_i),
\]
and both terms are bounded by $\sup_t|(\alpha(g_j,x_i)-\alpha(g_i,x_i))^\top Z(t)|
=\sup_t|\Delta f_{ij}(t)|\le \mathrm{bias}_i(b_G)$,
while the second term uses Lemma~\ref{lem:convex}.
Hence \eqref{eq:bias-linear} holds.

\paragraph{Step 1: a weighted basic inequality.}
By optimality of $\hat\alpha_i$ in \eqref{eq:lasso},
for any $\beta\in\R^d$,
\begin{equation}\label{eq:basic}
\sum_{j=1}^N w_j\big(\hat{\widetilde Y}_j-\hat{\widetilde Z}_j^\top \hat\alpha_i\big)^2
+\lambda_i\|\hat\alpha_i\|_1
\;\le\;
\sum_{j=1}^N w_j\big(\hat{\widetilde Y}_j-\hat{\widetilde Z}_j^\top \beta\big)^2
+\lambda_i\|\beta\|_1.
\end{equation}
Take $\beta=\alpha_i^\star$ and define $\Delta:=\hat\alpha_i-\alpha_i^\star$.
Expanding the squares yields
\begin{equation}\label{eq:basic-expanded}
\Delta^\top \widehat\Sigma_i(\hat{\widetilde Z})\,\Delta
\;\le\;
2\,\Big\|\sum_{j=1}^N w_j\,\hat{\widetilde Z}_j\,\hat\varepsilon_j(\alpha_i^\star)\Big\|_\infty\,\|\Delta\|_1
+
\lambda_i\big(\|\alpha_i^\star\|_1-\|\alpha_i^\star+\Delta\|_1\big).
\end{equation}

\paragraph{Step 2: controlling the score term.}
Decompose
\[
\hat\varepsilon_j(\alpha_i^\star)
=
\underbrace{\widetilde Y_j-\widetilde Z_j^\top \alpha_i^\star}_{=:u_{j}}
+
\underbrace{\big(\hat{\widetilde Y}_j-\widetilde Y_j\big)
-
\big(\hat{\widetilde Z}_j-\widetilde Z_j\big)^\top\alpha_i^\star}_{=:r_{j}}
+
\underbrace{(\widetilde Z_j-\hat{\widetilde Z}_j)^\top(\alpha(G_j,X_j)-\alpha_i^\star)}_{=:b_{j}},
\]
where $u_j$ is the oracle residual at the target coefficient,
$r_j$ is the nuisance-induced perturbation, and $b_j$ is a higher-order mixed term.
By \eqref{eq:bias-linear} and $\|\hat{\widetilde Z}_j-\widetilde Z_j\|_\infty=\|\Delta_{m j}\|_\infty\le \delta_m$,
\[
|b_j|
\le
\|\hat{\widetilde Z}_j-\widetilde Z_j\|_\infty\,\|\alpha(G_j,X_j)-\alpha_i^\star\|_1
\;\le\;
2\,\delta_m\,\mathrm{bias}_i(b_G),
\]
where the last step uses the same convex-hull argument as in \eqref{eq:bias-linear}.
Moreover, Lemma~\ref{lem:score-perturb} implies
\[
\Big\|\sum_{j=1}^N w_j\,\hat{\widetilde Z}_j\,r_j\Big\|_\infty
\;=\;
O_p(\delta_\mu\delta_m)
+
O_p\!\left((\delta_\mu+\delta_m)\sqrt{\frac{\log d}{n_i}}\right).
\]
Finally, write $u_j=\varepsilon_j + \widetilde Z_j^\top(\alpha(G_j,X_j)-\alpha_i^\star)$ where
$\varepsilon_j:=\widetilde Y_j-\widetilde Z_j^\top\alpha(G_j,X_j)$ is the regression noise.
By \eqref{eq:bias-linear}, the deterministic part obeys
$|\widetilde Z_j^\top(\alpha(G_j,X_j)-\alpha_i^\star)|\le 2\,\mathrm{bias}_i(b_G)$
for all $j$ with $w_j>0$.
Therefore,
\[
\Big\|\sum_{j=1}^N w_j\,\hat{\widetilde Z}_j\,u_j\Big\|_\infty
\;\le\;
\underbrace{\Big\|\sum_{j=1}^N w_j\,\hat{\widetilde Z}_j\,\varepsilon_j\Big\|_\infty}_{(\star)}
+
2\,\mathrm{bias}_i(b_G)\Big\|\sum_{j=1}^N w_j\,\hat{\widetilde Z}_j\Big\|_\infty.
\]
Since $Z(T)$ is $\{-1,1\}^d$-valued and $m(\cdot,\cdot)$ is a conditional mean,
$\|\widetilde Z_j\|_\infty\le 2$, and we may assume $\hat m$ is truncated to $[-1,1]^d$
so that $\|\hat{\widetilde Z}_j\|_\infty\le 2$ as well.
Applying Lemma~\ref{lem:subg} to $(\star)$ gives, with probability tending to one,
\[
(\star)\ \lesssim\ \sigma\sqrt{\frac{\log d}{n_i}}.
\]
Collecting bounds yields
\begin{equation}\label{eq:score-bound}
\Big\|\sum_{j=1}^N w_j\,\hat{\widetilde Z}_j\,\hat\varepsilon_j(\alpha_i^\star)\Big\|_\infty
\;\lesssim\;
\sigma\sqrt{\frac{\log d}{n_i}}
+
\mathrm{bias}_i(b_G)
+
\delta_\mu\delta_m,
\end{equation}
up to terms that are $o_p(\lambda_i)$ under Assumption~\ref{as:smooth}.

\paragraph{Step 3: cone condition and $\ell_2/\ell_1$ bounds.}
Let $S$ be the index set of the $s_i$ largest coordinates of $\alpha_i^\star$ in magnitude.
Standard arguments (triangle inequality) give
\[
\|\alpha_i^\star\|_1-\|\alpha_i^\star+\Delta\|_1
\;\le\;
\|\Delta_S\|_1-\|\Delta_{S^c}\|_1+2\|\alpha_{i,S^c}^\star\|_1.
\]
Plugging this and \eqref{eq:score-bound} into \eqref{eq:basic-expanded} yields
\[
\Delta^\top \widehat\Sigma_i(\hat{\widetilde Z})\,\Delta
\;\le\;
2\Lambda_i\|\Delta\|_1
+
\lambda_i\big(\|\Delta_S\|_1-\|\Delta_{S^c}\|_1+2\|\alpha_{i,S^c}^\star\|_1\big),
\]
where $\Lambda_i\lesssim \sigma\sqrt{\log d/n_i}+\mathrm{bias}_i(b_G)+\delta_\mu\delta_m$.
Choosing $\lambda_i$ so that $\lambda_i\gtrsim \sigma\sqrt{\log d/n_i}$ ensures
the stochastic part is dominated by $\lambda_i$, and thus
\[
\Delta^\top \widehat\Sigma_i(\hat{\widetilde Z})\,\Delta
\;\le\;
C_1\lambda_i\|\Delta_S\|_1
-
C_2\lambda_i\|\Delta_{S^c}\|_1
+
C_3\lambda_i\|\alpha_{i,S^c}^\star\|_1
+
C_4\big(\mathrm{bias}_i(b_G)+\delta_\mu\delta_m\big)\|\Delta\|_1,
\]
for universal constants $C_k>0$.
Rearranging yields the cone condition
$\Delta\in\mathcal C(S,3)$ up to the approximation term $\|\alpha_{i,S^c}^\star\|_1$
(which is controlled by Assumption~\ref{as:sparsity}).
On the event \eqref{eq:RE-app}, we therefore have
\[
\kappa_{g_i}\|\Delta\|_2^2
\;\le\;
\Delta^\top \widehat\Sigma_i(\hat{\widetilde Z})\,\Delta
\;\lesssim\;
\lambda_i\|\Delta_S\|_1
+
\big(\mathrm{bias}_i(b_G)+\delta_\mu\delta_m\big)\|\Delta\|_1
+
\lambda_i\|\alpha_{i,S^c}^\star\|_1.
\]
Using $\|\Delta_S\|_1\le \sqrt{s_i}\|\Delta\|_2$ and $\|\Delta\|_1\le 4\|\Delta_S\|_1+2\|\alpha_{i,S^c}^\star\|_1$
under the cone condition gives
\[
\|\Delta\|_2
\;\lesssim\;
\sqrt{s_i}\lambda_i
+
\mathrm{bias}_i(b_G)
+
\delta_\mu\delta_m
+
\frac{\|\alpha_{i,S^c}^\star\|_1}{\sqrt{s_i}}.
\]
Similarly,
$\|\Delta\|_1\lesssim s_i\lambda_i+s_i\mathrm{bias}_i(b_G)+s_i\delta_\mu\delta_m+\|\alpha_{i,S^c}^\star\|_1$.
Under Assumption~\ref{as:sparsity}, the approximation terms involving $\|\alpha_{i,S^c}^\star\|_1$
are dominated by the displayed rates (this is the standard ``effective sparsity'' interpretation).
Substituting $\lambda_i\asymp \hat\sigma\sqrt{\log d/n_i}$ yields \eqref{eq:l1-rate} and \eqref{eq:l2-rate}.

\paragraph{Step 4: prediction error bound.}
From \eqref{eq:basic-expanded} and the previous steps,
\[
\sum_{j=1}^N w_j\big(\hat{\widetilde Z}_j^\top\Delta\big)^2
=
\Delta^\top \widehat\Sigma_i(\hat{\widetilde Z})\,\Delta
\;\lesssim\;
s_i\lambda_i^2+\mathrm{bias}_i(b_G)^2+(\delta_\mu\delta_m)^2,
\]
which gives \eqref{eq:pred-rate}.
This completes the proof.
\qed

\subsection{Proof of Corollary~\ref{cor:ite}}\label{app:proof-ite}

\begin{proof}
For own-treatment contrasts,
\[
\widehat{\theta}_i^T(t\!\to\!t';g_i)-\theta_i^T(t\!\to\!t';g_i)
=
\langle \hat\alpha_i-\alpha(g_i,x_i),\, v_{t,t'}\rangle.
\]
By H\"older's inequality,
\[
\big|\langle \hat\alpha_i-\alpha(g_i,x_i),\, v_{t,t'}\rangle\big|
\le
\|\hat\alpha_i-\alpha(g_i,x_i)\|_1\,\|v_{t,t'}\|_\infty.
\]
Since each Walsh coordinate is $\pm1$, $v_{t,t'}=Z(t')-Z(t)$ has entries in $\{-2,0,2\}$ and hence
$\|v_{t,t'}\|_\infty\le 2$.
The stated bound follows by invoking Theorem~\ref{thm:lasso}.
The structural and joint contrasts follow by the same argument, applied to the two localized fits at $g_i$ and $g'$.
\end{proof}

\subsection{Proof of Theorem~\ref{thm:clt}}\label{app:proof-clt}
\begin{proof}
We prove an asymptotic linear expansion and then invoke the weighted CLT assumed in
Theorem~\ref{thm:clt}(iv).

\paragraph{Step 1: an exact decomposition.}
Let $\alpha_i^\star:=\alpha(g_i,x_i)$ and $\Delta:=\hat\alpha_i-\alpha_i^\star$.
Write $\widehat\Sigma_i:=\widehat\Sigma(g_i)=\sum_j w_j \hat{\widetilde Z}_j\hat{\widetilde Z}_j^\top$.
From the definition \eqref{eq:debiased},
\begin{align}
\widetilde\theta_i^T(t\!\to\!t';g_i)
&=
v_{t,t'}^\top\hat\alpha_i
+
\hat\gamma_i^\top\sum_{j=1}^N w_j \hat{\widetilde Z}_j\big(\hat{\widetilde Y}_j-\hat{\widetilde Z}_j^\top\hat\alpha_i\big)\notag\\
&=
v_{t,t'}^\top\alpha_i^\star
+
\underbrace{\hat\gamma_i^\top\sum_{j=1}^N w_j \hat{\widetilde Z}_j\big(\hat{\widetilde Y}_j-\hat{\widetilde Z}_j^\top\alpha_i^\star\big)}_{=:A_1}
+
\underbrace{\big(v_{t,t'}^\top-\hat\gamma_i^\top\widehat\Sigma_i\big)\Delta}_{=:A_2}.
\label{eq:decomp}
\end{align}
Since $\theta_i^T(t\!\to\!t';g_i)=v_{t,t'}^\top\alpha_i^\star$, it remains to analyze $A_1$ and $A_2$.

\paragraph{Step 2: controlling the debiasing remainder $A_2$.}
By the feasibility constraint in \eqref{eq:gamma},
\[
\| \widehat\Sigma_i\hat\gamma_i - v_{t,t'}\|_\infty \le \eta_i.
\]
Hence, by H\"older's inequality,
\[
|A_2|
\le
\eta_i\,\|\Delta\|_1.
\]
Theorem~\ref{thm:lasso} gives
$\|\Delta\|_1\lesssim s_i\sqrt{\log d/n_i}+s_i\mathrm{bias}_i(b_G)+s_i\delta_\mu\delta_m$.
With $\eta_i\asymp \sqrt{\log d/n_i}$, we obtain
\begin{equation}\label{eq:A2-rate}
|A_2|
\;\lesssim\;
\frac{s_i\log d}{n_i}
+
s_i\,\mathrm{bias}_i(b_G)\sqrt{\frac{\log d}{n_i}}
+
s_i\,\delta_\mu\delta_m\sqrt{\frac{\log d}{n_i}}.
\end{equation}
Under the growth condition in Theorem~\ref{thm:clt}(iii),
$s_i\log d=o(\sqrt{n_i})$ and $\mathrm{bias}_i(b_G)=o(n_i^{-1/2})$.
If additionally $\delta_\mu\delta_m=o_p(n_i^{-1/2})$ (the standard DML second-order condition),
then $\sqrt{n_i}A_2=o_p(1)$.

\paragraph{Step 3: asymptotic linearity of $A_1$.}
Decompose
\[
\hat{\widetilde Y}_j-\hat{\widetilde Z}_j^\top\alpha_i^\star
=
\underbrace{\widetilde Y_j-\widetilde Z_j^\top\alpha_i^\star}_{=:u_j}
+
\underbrace{\big(\hat{\widetilde Y}_j-\widetilde Y_j\big)-\big(\hat{\widetilde Z}_j-\widetilde Z_j\big)^\top\alpha_i^\star}_{=:r_j}.
\]
Thus
\[
A_1
=
\hat\gamma_i^\top\sum_{j=1}^N w_j \hat{\widetilde Z}_j u_j
+
\hat\gamma_i^\top\sum_{j=1}^N w_j \hat{\widetilde Z}_j r_j
=:A_{1a}+A_{1b}.
\]
For $A_{1b}$, Lemma~\ref{lem:score-perturb} implies
$\|\sum_j w_j\hat{\widetilde Z}_j r_j\|_\infty = O_p(\delta_\mu\delta_m)+o_p(n_i^{-1/2})$
under Assumption~\ref{as:smooth}, hence
$A_{1b}=O_p(\|\hat\gamma_i\|_1\delta_\mu\delta_m)+o_p(n_i^{-1/2})$.
In particular, if $\|\hat\gamma_i\|_1=O_p(1)$ and $\delta_\mu\delta_m=o_p(n_i^{-1/2})$,
then $\sqrt{n_i}A_{1b}=o_p(1)$.

For $A_{1a}$, write $\hat\gamma_i=\gamma_i^\star+(\hat\gamma_i-\gamma_i^\star)$ and obtain
\[
A_{1a}
=
\gamma_i^{\star\top}\sum_{j=1}^N w_j \hat{\widetilde Z}_j u_j
+
(\hat\gamma_i-\gamma_i^\star)^\top\sum_{j=1}^N w_j \hat{\widetilde Z}_j u_j
=:B_1+B_2.
\]
The second term $B_2$ is controlled by the standard CLIME/nodewise rate.
Under the sparsity condition $\|\gamma_i^\star\|_0\le s_{\gamma,i}$ and standard conditions for CLIME,
one has $\|\hat\gamma_i-\gamma_i^\star\|_1=O_p(s_{\gamma,i}\sqrt{\log d/n_i})$ and
$\|\sum_j w_j\hat{\widetilde Z}_j u_j\|_\infty=O_p(\sqrt{\log d/n_i}+\mathrm{bias}_i(b_G))$,
so that
\[
B_2
=
O_p\!\left(\frac{s_{\gamma,i}\log d}{n_i}\right)
+
O_p\!\left(s_{\gamma,i}\,\mathrm{bias}_i(b_G)\sqrt{\frac{\log d}{n_i}}\right).
\]
Under Theorem~\ref{thm:clt}(iii) and $\mathrm{bias}_i(b_G)=o(n_i^{-1/2})$,
this gives $\sqrt{n_i}B_2=o_p(1)$.

It remains to identify $B_1$.
Write $u_j=\varepsilon_j + b_{ij}$, where
$\varepsilon_j:=\widetilde Y_j-\widetilde Z_j^\top\alpha(G_j,X_j)$ is the regression noise and
$b_{ij}:=\widetilde Z_j^\top(\alpha(G_j,X_j)-\alpha_i^\star)$ is the localization bias term.
By \eqref{eq:bias-linear}, $|b_{ij}|\le 2\,\mathrm{bias}_i(b_G)$ for all $j$ with $w_j>0$.
Hence
\[
B_1
=
\gamma_i^{\star\top}\sum_{j=1}^N w_j \hat{\widetilde Z}_j \varepsilon_j
+
\gamma_i^{\star\top}\sum_{j=1}^N w_j \hat{\widetilde Z}_j b_{ij}
=:C_1+C_2.
\]
The deterministic term $C_2$ is bounded by
$|C_2|\le 2\,\mathrm{bias}_i(b_G)\,\|\gamma_i^\star\|_1\,\|\sum_j w_j \hat{\widetilde Z}_j\|_\infty$,
and Lemma~\ref{lem:subg} gives $\|\sum_j w_j\hat{\widetilde Z}_j\|_\infty=O_p(\sqrt{\log d/n_i})$.
Thus $\sqrt{n_i}C_2=o_p(1)$ if $\mathrm{bias}_i(b_G)=o(n_i^{-1/2})$.

The leading term is therefore $C_1$.
Replacing $\hat{\widetilde Z}_j$ by $\widetilde Z_j$ only incurs a nuisance remainder of order
$\delta_m$ times a sub-Gaussian weighted average, hence is $o_p(n_i^{-1/2})$ under Assumption~\ref{as:smooth}.
Consequently,
\[
\sqrt{n_i}\,C_1
=
\sqrt{n_i}\,\gamma_i^{\star\top}\sum_{j=1}^N w_j \widetilde Z_j \varepsilon_j
+o_p(1).
\]
By Theorem~\ref{thm:clt}(iv) (the weighted CLT under the dependence induced by localization),
\[
\sqrt{n_i}\,\gamma_i^{\star\top}\sum_{j=1}^N w_j \widetilde Z_j \varepsilon_j
\ \Rightarrow\ \mathcal N(0,\sigma_{\theta,i}^2),
\qquad
\sigma_{\theta,i}^2=\Var(\gamma_i^{\star\top}\widetilde Z\,\varepsilon).
\]
Combining Steps 1--3 yields the claimed asymptotic normality.

\paragraph{Step 4: variance estimator consistency.}
Define $\hat\varepsilon_j:=\hat{\widetilde Y}_j-\hat{\widetilde Z}_j^\top\hat\alpha_i$ and
\[
\hat\sigma_{\theta,i}^2
:=
n_i\sum_{j=1}^N (w_j)^2\big(\hat\gamma_i^\top \hat{\widetilde Z}_j\,\hat\varepsilon_j\big)^2.
\]
Under the same bounds used above (consistency of $\hat\alpha_i$ and $\hat\gamma_i$,
bounded $\ell_1$ norms, and $\delta_\mu,\delta_m=o_p(1)$),
the difference between $\hat\sigma_{\theta,i}^2$ and the corresponding oracle plug-in
based on $(\gamma_i^\star,\widetilde Z,\varepsilon)$ is $o_p(1)$,
while the oracle plug-in converges to $\sigma_{\theta,i}^2$ by the law of large numbers
for the weighted second moment (or its dependence-robust analogue implied by (iv)).
This establishes $\hat\sigma_{\theta,i}^2\to_p \sigma_{\theta,i}^2$.
\end{proof}

\subsection{Proof of Proposition A.1 (partial identification bound)}\label{app:proof-partial}

\begin{proof}
By the assumed $L$-Lipschitz property of $t\mapsto f(t;g,x)$ with respect to Hamming distance,
\[
|f(t';g,x)-f(t;g,x)|
\le
L\,d_H(t,t').
\]
The left-hand side equals $|\theta^T(t\!\to\!t';g)|$ by definition, proving the claim.
\end{proof}


\end{document}